\documentclass[aps,pra,reprint,superscriptaddress, longbibliography]{revtex4-2}
\usepackage{physics}
\usepackage{amsmath,bm}
\usepackage{bbold}
\usepackage{comment}
\usepackage[english]{babel}
\usepackage{amsmath}
\usepackage{graphicx}
\usepackage[colorlinks=true, allcolors=blue]{hyperref}
\usepackage{booktabs}

\usepackage[dvipsnames]{xcolor}
\usepackage{amsthm}
\newtheorem{theorem}{Theorem}
\newtheorem{lemma}{Lemma}
\usepackage{algorithm}
\usepackage{algorithmic}
\usepackage{booktabs}
\usepackage{bm}
\usepackage{tikz}
\usepackage{amssymb}
\usepackage{subcaption}
\usepackage[toc,page]{appendix}
\usepackage{mathtools}
\usepackage{enumerate}
\definecolor{darkpowderblue}{rgb}{0.0, 0.2, 0.7}
\definecolor{color1}{HTML}{5ec962} 
\definecolor{color2}{HTML}{21918c} 
\definecolor{color3}{HTML}{3b528b} 
\usepackage{hyperref}
\hypersetup{
    colorlinks=true,
    linkcolor=darkpowderblue,
    linktocpage=true,
    filecolor=magenta,
    citecolor=darkpowderblue,
    urlcolor=cyan
    }
\usepackage{pifont}
\usepackage{placeins}
\usepackage{tikz}
\usetikzlibrary{positioning}

\begin{document}

\title{Solving Helmholtz problems with finite elements on a quantum annealer}

\newcommand{\AffilMontefiore}{Department of Electrical Engineering and Computer Science, Institut Montefiore, University of Li\`ege, 4000 Li\`ege, Belgium}
\newcommand{\AffilCESAM}{Institut de Physique Nucl\'eaire, Atomique et de Spectroscopie, CESAM, University of Li\`ege, 4000 Li\`ege, Belgium}

\author{Arnaud Rémi}
\affiliation{\AffilMontefiore}

\author{François Damanet}
\affiliation{\AffilCESAM}

\author{Christophe Geuzaine}
\affiliation{\AffilMontefiore}

\begin{abstract}
    Solving Helmholtz problems using finite elements leads to the resolution of a linear system which is challenging to solve for classical computers. In this paper, we investigate how quantum annealers could address this challenge. We first express the linear system arising from the Helmholtz problem as a generalized eigenvalue problem (gEVP). The obtained gEVP is mapped into quadratic unconstrained binary optimization problems (QUBOs) which we solve using an adaptive quantum annealing eigensolver (AQAE) and its classical equivalent. We identify two key parameters in the success of AQAE for solving Helmholtz problems: the system condition number and the integrated control errors (ICE) in the quantum hardware. Our results show that a large system condition number implies a finer discretization grid for AQAE to converge, leading to a variable overhead, and that AQAE is either tolerant or not with respect to ICE depending on the gEVP. Finally, we establish lower bounds on the annealing time, narrowing the possibility of a quantum advantage for solving Helmholtz problems.
\end{abstract}

\maketitle

\section{Introduction}

Over the last decades, quantum computing has emerged as a crucial field to invest in~\cite{Pfaendler2024}, given its potential to solve computational problems that are intractable for classical computers~\cite{Lau2022, Liu2022, Golestan2023}. On the hardware side, two classes of platforms hold their own: i) universal gate-based quantum computers, which make it possible to perform arbitrary calculations, and ii) quantum annealers~\cite{Hauke2020}, which allow for solving specific optimization problems that can be mapped into Ising [or equivalently quadratic unconstrained binary optimization (QUBO)] problems. While quantum annealers seem more restrictive than gate-based quantum computers, they are in fact equivalent~\cite{aharonov2004}, although depending on the problem to solve one approach is usually more natural than the other. In this context, as Ising problems belong to the NP-hard computational complexity class~\cite{lucas2014ising}, quantum annealers are well-suited for solving a wide range of real-life problems, ranging from protein folding with health applications to finance, linear algebra or logistics.

However, despite the fast development of quantum hardware, the current numbers of qubits available as well as their fragility against noise have hindered, up to now, the achievement of practical quantum advantages~\cite{Herrmann2023,Daley2022, zimboras2025myths}. Yet, anticipating the advent of large-scale fault-tolerant quantum computers, many quantum algorithms with applications in various fields~\cite{arnault2024} have been developed.

On the application side, while quantum simulation remains one of the primary targets for quantum algorithms~\cite{lloyd1996universal, daley2022practical}, the scope of their application has expanded to encompass a broader range of partial dif  ferential equations (PDEs)~\cite{costa2019quantum, jin2022quantum, jin2023quantum}. In this study, we explore the near-term potential of quantum annealers for solving Helmholtz problems, a class of PDEs with wide-ranging applications in fields such as electromagnetism, acoustics and seismology. Classically, Helmholtz problems are addressed using discretization techniques such as the finite element method, which reduces the problem to solving a system of linear equations. The solution of these linear systems is however very challenging for current numerical solvers. One of the reasons is that the discretization of Helmholtz equations often results in indefinite systems, which are notoriously difficult for classical iterative solvers to handle efficiently~\cite{ernst2011difficult}. Furthermore, high-order finite element methods are commonly employed to accurately capture high-frequency oscillatory behavior, which generally leads to even more challenging linear systems for numerical solvers~\cite{eisentrager2020condition}. In this work, we investigate the performance and viability of quantum annealing-based approaches in addressing these challenges, offering insight into their potential advantages and limitations compared to classical methods.

Our paper is organized as follows. In Section II, we introduce Helmholtz problems, as well as their finite element formulations. We then introduce variational formulations of eigenvalue problems before showing how solving Helmholtz problems can be equivalent to solving generalized eigenvalue problems. In Section III, we present background on quantum annealing and we briefly overview the QAE and AQAE algorithms. In Section IV, we show our results and discuss them. Finally, the main conclusions are summarized in Section V.

\section{Helmholtz problem}

In this section, a one-dimensional Helmholtz problem with homogeneous Dirichlet and Neumann boundary conditions is introduced.
Two subproblems will be distinguished: i) the homogeneous case where solutions are the eigenfrequencies and the eigenmodes of the problem; ii) the nonhomogeneous case where the solution is the response of the system to a given excitation.

\subsection{One-dimensional Helmholtz problem}
The strong formulation of the one-dimensional Helmholtz problem on the interval $\Omega=\{x: 0<x<1\}$ with respectively homogeneous Dirichlet and Neumann boundary conditions on the left and on the right reads
\begin{equation}
    \begin{cases}
        \phi''(x) + k^2(x) \phi(x) = f(x), \quad x\in \Omega,\\
        \phi(0) = \phi'(1) = 0,
    \end{cases}
    \label{eq:hermitian_problem}
\end{equation}
where $\phi(x)$ is the system response, $k(x)=\omega/c(x)$ is the wavenumber with $c(x)$ the wave velocity and $\omega$ the angular frequency, $f(x)$ is the source term, and $\cdot'$ denotes the derivative with respect to $x$. In the following, we consider a heterogeneous domain consisting of vacuum in the left half (for $0 \leq x <1/2$, $k(x,\omega) = k_0 = \omega/c_0$) and $\mathrm{SiO}_2$ in the right half (for $1/2 \leq x \leq 1$, $k(x,\omega) = k_{\mathrm{SiO}_2} = \sqrt{3.9}k_0$), which could be viewed as a toy model for optical fibers. A discrete form of \eqref{eq:hermitian_problem} can be obtained using the finite element method.

\subsubsection{Finite element method}
The finite element method combines a weak formulation of the original PDE and a Galerkin projection using suitable basis functions.
\paragraph{Weak formulation.} To obtain a weak formulation of \eqref{eq:hermitian_problem}, the PDE is multiplied by test functions $v$ in  $H^1_D:=\{ f(x)\in H^1(\Omega): f(0) = 0\}$ and integrated by parts. It consists in finding $\phi(x)\in H^1_D$ such that
\begin{gather*}
    \int_0^1 \phi''v \,\mathrm{d}x + \int_0^1 k^2\phi v \, \mathrm{d}x = \int_0^1 fv \,\mathrm{d}x\\
\begin{split}
    \Leftrightarrow \quad -\int_0^1 \phi' v' \mathrm{d}x + \phi'(1)v(1) - \phi'(0)v(0) + \int_0^1k^2\phi v\, \mathrm{d}x\\= \int_0^1 fv \,\mathrm{d}x
\end{split}
\end{gather*}
holds for all test functions $v\in H^1_D$.
The boundary terms vanish since $v(0)=0$ and $\phi'(1)=0$, and the weak form simplifies into finding $\phi(x)\in H^1_D$ such that
\begin{equation}
    -\int_0^1 \phi' v' \mathrm{d}x + \int_0^1k^2\phi v\, \mathrm{d}x= \int_0^1 fv \,\mathrm{d}x, \quad\forall v\in H^1_D.
    \label{eq:weak_form}
\end{equation}

\paragraph{Galerkin projection.} A Galerkin projection is a projection of the solution of the original PDE onto a finite-dimensional subspace spanned by a family of basis functions $\{\varphi_j(x)\}_{j=0}^{N_{\mathrm{DOF}}-1}$. The solution $\phi$ of the Helmholtz problem is hence approximated by
\begin{equation}
    \phi(x) \approx \sum_{j=0}^{N_{\mathrm{DOF}}-1} \phi_j \varphi_j(x),
    \label{eq:galerkin_proj}
\end{equation}
with $\{\phi_j\}_{j=0}^{N_{\mathrm{DOF}}-1}$ a set of unknown coefficients.
By substituting \eqref{eq:galerkin_proj} into \eqref{eq:weak_form}, and by using the same basis functions as test functions, the weak formulation can be written in the following discrete form: find the coefficients $\phi_j$ such that
\begin{multline}
    \sum_{j=0}^{N_{\mathrm{DOF}}-1} \int_0^1 (-\varphi_j'\varphi_i')\, \mathrm{d}x \;\phi_j + \sum_{j=0}^{N_{\mathrm{DOF}}-1} \int_0^1 k^2 \varphi_j \varphi_i \,\mathrm{d}x \;\phi_j \\= \int_0^1 f\varphi_i \,\mathrm{d}x, \quad\forall i.
\end{multline}
This can be written in matrix form
\begin{equation}
    (K+M)\bm \phi = \bm f, \quad
    \begin{cases}
        K_{ij} = \int_0^1 (-\varphi_i'\varphi_j') \, \mathrm{d}x,\\
        M_{ij} = \int_0^1 k^2 \varphi_i \varphi_j \,\mathrm{d}x,\\
        f_j = \int_0^1 f\varphi_j(x) \,\mathrm{d}x.
    \end{cases}
    \label{eq:linear_system}
\end{equation}
\paragraph{Finite elements.}
The spatial domain is discretized into $N$ non-overlapping adjacent finite elements $\{E_\alpha\}$. Each element is associated with $p+1$ nodes $\{x_0^{(\alpha)}, x_1^{(\alpha)}, \dots, x_{p}^{(\alpha)}\}$, which correspond to interpolation points within the element. The chosen family of basis functions for this work are Lagrange polynomials, defined as:
\begin{equation}
\varphi_j^{(\alpha)}(x) = \prod_{\substack{i=0 \\ i \neq j}}^{p} \frac{x - x_i^{(\alpha)}}{x_j^{(\alpha)} - x_i^{(\alpha)}},
\end{equation}
where $\varphi_j^{(\alpha)}(x)$ is the $j$-th basis function associated with node $x_j^{(\alpha)}$, which satisfies $\varphi_j^{(\alpha)}(x_i^{(\alpha)}) = \delta_{ji}$. Within each element $E_\alpha$, the nodes $x_i^{(\alpha)}$ are the Gauss-Lobatto-Legendre (GLL) nodes. The order $p$ of the basis functions as well as the number of elements $N$ are crucial regarding the discretization errors, but also regarding the system condition number. For sufficiently smooth problems, the discretization error of finite elements with GLL nodes is $\mathcal{O}(1/N)^{p+1}$, while the condition number is quadratic in $N$ and cubic in $p$~\cite{eisentrager2020condition}.

\subsection{Variational formulation of eigenvalue problems}
In this subsection, we introduce how generalized eigenvalue problems (gEVPs) can be expressed as optimization problems, and what the conditions are for a gEVP to have such a variational formulation. Let us consider the following gEVP
\begin{equation}\label{gEVP}
    H\bm \phi = \lambda M \bm \phi,
\end{equation}
where $H$ is a Hermitian matrix of dimension $N_\mathrm{DOF} \times N_\mathrm{DOF}$, $M$ positive definite, and $\lambda$ and $\bm \phi$ the eigenvalues and eigenvectors of the gEVP, respectively. The solutions of such problems are pairs of eigenvalues and eigenvectors labeled $\{\lambda_n, \bm{\phi}_n \} = \{ \lambda, \bm{\phi}\}$ ($n = 0, \dotsc, N_{\mathrm{DOF}}-1$).
The min-max theorem allows to formulate the search of eigenpairs of an eigenvalue problem (standard or generalized) as a variational problem of the form (see App.~\ref{app:minmaxthm})
\begin{align}\label{syst:minmax1}
    &\lambda_n = \, \displaystyle\min_{\bm{\phi} \in \mathcal{S}_{n}} \frac{\bm{\phi}^\dag H \bm{\phi}}{\bm{\phi}^\dag M \bm{\phi}},\\
    &\bm \phi_n = \arg  \, \displaystyle\min_{\bm{\phi} \in \mathcal{S}_{n}} \frac{\bm{\phi}^\dag H \bm{\phi}}{\bm{\phi}^\dag M \bm{\phi}},
\label{syst:minmax2}
\end{align}
with $\cdot^\dag$ the conjugate transpose operator, and $\mathcal{S}_n = \mathrm{span}(\bm \phi_n, \cdots, \bm \phi_{N_{\mathrm{DOF}}-1})$. It can be shown that solving Eqs.\eqref{syst:minmax1}-\eqref{syst:minmax2} is equivalent to solving
\begin{align}
    \lambda_n &= \, \displaystyle\min_{\bm{\phi} \in \mathcal{S}_{n}} {\bm{\phi}^\dag H \bm{\phi}} \quad :\quad \bm \phi^\dag M \bm \phi = 1,\\
    \bm \phi_n &= \arg \, \displaystyle\min_{\bm{\phi} \in \mathcal{S}_{n}} {\bm{\phi}^\dag H \bm{\phi}}\quad :\quad \bm \phi^\dag M \bm \phi = 1.
\end{align}
 or  \begin{align}\label{eq:quadraticminmax1}
    \lambda_n &= \, \displaystyle\min_{\lambda \in \mathbb{R},\, \bm{\phi} \in \mathcal{S}_{n}} \bm{\phi}^\dag H \bm{\phi} - \lambda (\bm{\phi}^\dag M \bm{\phi} - 1),\\
    \bm \phi_n &= \arg \,
    \displaystyle\min_{\lambda \in \mathbb{R},\,\bm{\phi} \in \mathcal{S}_{n}} \bm{\phi}^\dag H \bm{\phi} - \lambda (\bm{\phi}^\dag M \bm{\phi} - 1),
    \label{eq:quadraticminmax2}
\end{align}
using the Lagrange multiplier method where $\lambda$ is the Lagrange multiplier. This formulation allows to compute all the eigenvalue-eigenvector pairs iteratively starting with $n=0$ and $\mathcal{S}_0 = \mathbb{C}^{N_{\mathrm{DOF}}}$. For $n>0$, the constraint $\bm \phi \in \mathcal{S}_n$ can be ensured by adding a projection term (see App.\ref{app:gevpqubomapping} for more details).

\subsection{Homogeneous problem}
The first problem that will be tackled below is the homogeneous Helmholtz problem. Its strong formulation writes
\begin{equation}
    \begin{cases}
        \phi''(x, \omega) + k^2(x,\omega) \phi(x, \omega) =0, \quad x\in \Omega,\\
        \phi(0)  = \phi'(1) = 0,
    \end{cases}
    \label{eq:homogeneous_problem}
\end{equation}
and has as solutions the eigenmodes with associated frequencies. The finite element discretization of this problem is given by
\begin{equation}
   \Tilde{K}\bm \phi = \lambda \Tilde{M}\bm \phi,\quad
       \begin{cases}
        \Tilde K_{ij} = \int_0^1 (-c^2 \varphi_i'\varphi_j') \, \mathrm{d}x,\\
        \Tilde M_{ij} = \int_0^1  \varphi_i \varphi_j \,\mathrm{d}x.
    \end{cases}
    \label{eq:homogeneous_gevp}
\end{equation}
which is a generalized eigenvalue problem for the ($\Tilde K, \Tilde M$) pair, with eigenvalue $\lambda$ related to the frequency through $\lambda = - \omega^2$. Using \eqref{eq:quadraticminmax1} and \eqref{eq:quadraticminmax2}, we can write the gEVP variational formulation as
\begin{align}
    \lambda_n &= \, \displaystyle\min_{\lambda \in \mathbb{R},\, \bm{\phi} \in \mathcal{S}_{n}} \bm{\phi}^\dag \Tilde{K} \bm{\phi} - \lambda (\bm{\phi}^\dag \Tilde{M} \bm{\phi} - 1),\\
    \bm \phi_n &= \arg \,
    \displaystyle\min_{\lambda \in \mathbb{R},\,\bm{\phi} \in \mathcal{S}_{n}} \bm{\phi}^\dag \Tilde{K} \bm{\phi} - \lambda (\bm{\phi}^\dag \Tilde{M} \bm{\phi} - 1).
\end{align}
\subsection{Non-homogeneous problem}
The second problem that will be solved is the non-homogeneous Helmholtz problem \eqref{eq:hermitian_problem} with $f(x) = \sin(2\pi x)$ and $k_0\in\{0, \pi, 2\pi\}$. Once discretized, it leads to the linear system \eqref{eq:linear_system}. By writing $\bm \phi = -\lambda \bm w/(\bm f^\dag \bm w)$, we obtain the following generalized eigenvalue problem:
\begin{equation}
    -\bm f \bm f^\dag \,\bm w= \lambda (K+M) \bm w
\end{equation}
which can be solved for $\lambda$ and $\bm w$. For the min-max theorem to hold, $(K+M)$ must be positive definite which may not be the case depending on the wavenumber $k$. Hence, for large $k$, the so-called normal equations have to be considered instead
\begin{equation}
    (K+M)^\dag (K+M) \bm \phi = (K+M)^\dag \bm f,
    \label{eq:helmholtz_normal}
\end{equation}
which will be written $A\bm \phi = \bm b$ with $A =  (K+M)^\dag (K+M)$ and $\bm b = (K+M)^\dag \bm f$ for simplicity. The normal equations lead to the following generalized eigenvalue problem:
\begin{equation}
    -\bm b \bm b^\dag \,\bm w= \lambda A \bm w,
    \label{eq:helmholtz_normal_gevp}
\end{equation}
where $\bm b \bm b^\dag$ is Hermitian and $A$ is positive definite. Here again, using \eqref{eq:quadraticminmax1} and \eqref{eq:quadraticminmax2}, we can write the gEVP variational formulation as
\begin{align}
    \lambda_n &= \, \displaystyle\min_{\lambda \in \mathbb{R},\, \bm{w} \in \mathcal{S}_{n}} \bm{w}^\dag (-\bm b \bm b^\dag) \bm{w} - \lambda (\bm{w}^\dag A \bm{w} - 1),\\
    \bm w_n &= \arg \,
    \displaystyle\min_{\lambda \in \mathbb{R},\,\bm{w} \in \mathcal{S}_{n}} \bm{w}^\dag  (-\bm b \bm b^\dag) \bm{w} - \lambda (\bm{w}^\dag A \bm{w} - 1).
\end{align}
The ground state $(\lambda_0,\bm w_0)$ of \eqref{eq:helmholtz_normal_gevp} is therefore related to the system response by $\bm \phi = -\lambda_0 \bm w_0 /\bm f^\dag \bm w_0$.

\section{Quantum annealing}\label{sec:QA}

Quantum annealing (QA) is a metaheuristic, i.e., a method to find (sufficiently good) solutions of optimization problems performed by quantum annealers. To do so, the qubits in the hardware are first initialised in a easy-to-prepare ground state of a Hamiltonian $H_0$. Then, their interactions are adiabatically adapted until reaching the values of the parameters $J_{kl}$ and $h_k$ of a classical Ising Hamiltonian $H_\mathrm{p}$ encoding the solution of the optimization problem of interest and of the form
\begin{equation}\label{Ising}
    H_\mathrm{p} = \sum_\alpha h_\alpha s_\alpha + \sum_{\alpha,\beta} J_{\alpha \beta} s_\alpha s_\beta,
\end{equation}
where $s_\alpha = \{-1,1\}$ are the spin variables stored in the qubits states. The time-dependent Hamiltonian governing the quantum annealing dynamics can be parametrized as follows
\begin{equation}\label{eq:dwave_hamiltonian}
    H(s) = A(s)H_0 + B(s)H_\mathrm{p},\quad s\in [0,1],
\end{equation}
where $s = t/t_\mathrm{a}$ is an adimensional parameter with $t_\mathrm{a}$ the total annealing time, and $A$ and $B$ time-dependent scheduling parameters with typically $A(0)\gg B(0)$ and $A(1)\ll B(1)$ so that $H(0) \propto H_0$ and $H(1) \propto H_\mathrm{p}$. According to the adiabatic quantum theorem, if the Hamiltonian $H(s)$ evolves sufficiently slowly, the quantum state of the system should remain in the instantaneous ground state of $H(s)$ for all times, thereby evolving towards the ground state of $H_\mathrm{p}$ for $s = 1$~\cite{sakurai1995modern, apolloni1989quantum}, i.e., to the spin configuration optimizing the problem, which can be extracted via direct measurement of the qubits in the computational basis at the end of the quantum annealing process. The minimum energy gap between the ground state and the first excited state is denoted by $\Delta$, and represents a key parameter to determine the annealing time $t_\mathrm{a}$, as typically one requires $t_\mathrm{a} \gg \Delta^{-2}$~\cite{sakurai1995modern, McGeoch2014}, although other studies suggested $t_\mathrm{a}\gg \Delta^{-3}$ \cite{jansen2007bounds} or $t_\mathrm{a}\gg \Delta^{-2}|\log{\Delta}|^{\beta}$~\cite{elgart2012note}, which ultimately depends on the specific form of $H(s)$. Nevertheless, the choice of the annealing time represents a possible source of errors. Another source is the decoherence induced by the coupling of the hardware to its environment, which actually suggests the existence of a tradeoff for choosing $t_\mathrm{a}$, as larger $t_\mathrm{a}$ implies stronger decoherence effects. Finally, other errors can come from a bad encoding of the values of $J_{kl}$ and $h_k$ in the hardware, as deviations from the right values might modify the ground state, or more generally from the approximations used to obtain the Ising Hamiltonian or from the binarization procedure of the variables when one has to deal with real values problems, as elaborated below in our examples.

Finding the ground state of an Ising Hamiltonian of the form~(\ref{Ising}) is strictly equivalent to finding the solution of a Quadratic Unconstrained Binary Optimization (QUBO) problem of the form
\begin{equation}\label{QUBO}
    \bm{q}_\star = \arg \min_{\bm{q}} \bm{q}^T Q \bm{q},\quad \bm{q} \in \{0,1\}^D,
\end{equation}
where the binary variables $q_\alpha$ that are the components of the vector $\bm{q}$ of dimension $D$ are simply related to the spin variables through $q_\alpha = (s_\alpha + 1)/2$, and where the $\star$ symbol denotes the optimized value of $\bm{q}$. Higher-order binary polynomial functions can be minimized using QA by expressing them as binary quadratic functions, although this requires the introduction of additional binary variables \cite{biamonte2008nonperturbative, babbush2014adiabatic}.
In the following, we will be interested in formulating the solution of an eigenvalue problem as a QUBO of the form~(\ref{QUBO}), keeping in mind that it is equivalent to the Ising model~(\ref{Ising}) natively implemented in the hardware.

\subsubsection{QUBO formulation of gEVPs}
The gEVP variational formulation given in \eqref{eq:quadraticminmax1}-\eqref{eq:quadraticminmax2} is an optimization problem that has continuous variables. A QUBO formulation can be derived by parameterizing each continuous variables with $D$ binary ones:
\begin{align}\label{eq:continuousbinarymapping}
    &\phi_j(\bm q_j) = \phi_{j,\mathrm{min}} + \bm b^T \bm q_j,\\&\mathrm{with}\,\quad\begin{cases}
        \bm q_j \in \{0,1\}^D,\\
        b_k = (\phi_{j,\mathrm{max}}-\phi_{j,\mathrm{min}}) 2^{-k}, \quad k=0,...,D-1.\nonumber
    \end{cases}
\end{align}
Injecting \eqref{eq:continuousbinarymapping} into \eqref{eq:quadraticminmax2} lead to the following QUBO formulation
\begin{equation}
    \bm q_\star = \arg \min_{\bm q} \bm q^T Q_H \bm q - \gamma \bm q^T Q_M \bm q,
    \label{eq:gevpqubo}
\end{equation}
where $\bm q \in \{0,1\}^{N\times D}$ is given by
\begin{equation*}
    \bm q^T = [\bm q_0^T, ...,\bm q_{N-1}^T]
\end{equation*}
and the derivation of $Q_H$ and $Q_M$ is in app.~\ref{app:gevpqubomapping}.

\subsubsection{Adaptive quantum annealer eigensolver}
The gEVP QUBO formulation \eqref{eq:gevpqubo} gives rise to the quantum annealer eigensolver (QAE) \cite{teplukhin2019calculation, teplukhin2020electronic, teplukhin2020solving, teplukhin2021computing, teplukhin2022sampling} which, by performing $N_\lambda$ bisection iterations to adjust the value of the lagrange parameter $\lambda$ (see Alg.~\ref{alg:QAE} for the detailed algorithm), allows the calculation of the ground state of a gEVP of the form (\ref{gEVP}), with a number-of-qubit limited precision. A quantum annealer is called to solve the QUBO of each iteration. The number of bisection iterations is only related to the final precision on $\lambda$. Several tunable parameters such as the annealing time $t_\mathrm{a}$ and the number of samples $\texttt{numreads}$ (i.e., the number of times each QUBOs are solved to later choose the best solution) are inputs to the quantum annealer. An adaptive version of the QAE named adaptive quantum annealer eigensolver (AQAE) has been proposed in Ref.~\cite{a20212} to circumvent the limited precision of the QAE by refining the bounds on which continuous variables are defined at each iteration $N_\delta$ times using the box algorithm~\cite{srivastava2019box} according to a geometric ratio $r$. The final precision of AQAE is user-defined, and directly related to the number of box-algorithm iterations ($N_\delta \propto \log_r(1/\varepsilon)$). The AQAE algorithm is detailed in Alg.\ref{alg:AQAE}. If a classical annealer (i.e., a simulated annealing algorithm~\cite{bertsimas1993simulated} performed by a classical computer) is used instead of the quantum annealer to solve the QUBOs, Alg.~\ref{alg:AQAE} becomes the adaptive classical annealer eigensolver (ACAE). In this paper both AQAE and ACAE will be used. As classical computational resources are currently inexpansive compared to quantum ones, \texttt{numreads} can be chosen higher for ACAE than for AQAE when it is needed. By doing so, ACAE results are considered as a reference solution that we can try to reproduce using AQAE.

\begin{algorithm}[H]
\caption{QAE algorithm}
\begin{algorithmic}[1]
\REQUIRE $H$, $M$, $N_\lambda$, $\lambda_\mathrm{min}$, $\lambda_\mathrm{max}$, $\bm \phi_\mathrm{min}$, $\bm \phi_\mathrm{max}$
\STATE Compute $Q_H$ and $Q_M$
\STATE $i' \leftarrow 0$
\WHILE{$i' < N_\lambda$}
    \STATE $\lambda \leftarrow (\lambda_\mathrm{min} + \lambda_\mathrm{max})/2$
    \STATE $\bm q_\star \leftarrow \arg \min_{\bm q} \bm q^T Q_H\bm q - \lambda\bm q^T Q_M\bm q$
    \STATE $\bm \phi_\star \leftarrow \bm \phi (\bm q_\star)$
    \IF{$\bm \phi_\star^T H \bm \phi_\star - \lambda  \bm \phi_\star^T M \bm \phi_\star \geq 0$}
        \STATE $\lambda_\mathrm{min} \leftarrow \lambda$
    \ELSE
        \STATE $\lambda_\mathrm{max} \leftarrow \lambda$
    \ENDIF
    \STATE $i' \leftarrow i' + 1$
\ENDWHILE
\STATE $\lambda_\star \leftarrow \lambda_\mathrm{max}$
\RETURN $\bm \phi_\star$, $\lambda_\star$
\end{algorithmic}
\label{alg:QAE}
\end{algorithm}

\begin{algorithm}[H]
\caption{AQAE algorithm}
\begin{algorithmic}[1]
\REQUIRE $H$, $M$, $N_\lambda$, $N_\delta$, $\lambda_\mathrm{min}$, $\lambda_\mathrm{max}$, $\bm \phi_\mathrm{min}$, $\bm \phi_\mathrm{max}$, $r$
\STATE $\bm \delta \leftarrow \bm \phi_\mathrm{max} - \bm \phi_\mathrm{min}$
\STATE $i \leftarrow 0$
\WHILE{$i < N_\delta$}
    \STATE Run Alg.~\ref{alg:QAE} with \{$H$, $M$, $N_\lambda$, $\lambda_\mathrm{min}$ ,$\lambda_\mathrm{max}$, $\bm\phi_\mathrm{min}$, $\bm \phi_\mathrm{max}$\} as inputs and get {$\bm \phi_\star^{(i)}$,$\lambda^{(i)}$}
    \STATE $\bm \delta \leftarrow r \bm \delta$ 
    \STATE $\bm \phi_\mathrm{min} \leftarrow \bm \phi_\star - \bm \delta/2$ 
    \STATE $\bm \phi_\mathrm{max} \leftarrow \bm \phi_\star + \bm \delta/2$ 
    \STATE $i \leftarrow i + 1$
\ENDWHILE
\STATE $\lambda_\star \leftarrow \lambda^{(N_\delta-1)}$
\STATE $\bm \phi_\star \leftarrow \bm \phi_\star^{(N_\delta-1)}$
\RETURN $\bm \phi_\star, \lambda_\star$
\end{algorithmic}
\label{alg:AQAE}
\end{algorithm}

\subsubsection{Solving partial differential equations on a quantum annealer} 
The possibility of solving partial differential equations on a quantum annealer has recently been investigated \cite{srivastava2019box, criado2022qade, raisuddin2022feqa, gaidai2022molecular, conley2023quantum, nguyen2024quantum, kudo2024annealing}. Most of these investigations are based on minimizing the residual of the linear system arising from the discretization of PDE, that is, from a linear system $A\bm x = \bm b$, minimizing the functional
\begin{equation}
    \mathcal{J}(\bm x) = ||A\bm x - \bm b||^2,
\end{equation}
which ultimately has to be expressed using binary variables with a parameterization such as \eqref{eq:continuousbinarymapping}. An issue of this approach is that writing the parameterization \eqref{eq:continuousbinarymapping} requires specifying bounds over the continuous variables $\bm x$, that is, prescribing $\{\bm x_\mathrm{min}, \bm x_\mathrm{max}\}$ and implying $\bm x \in [\bm x_\mathrm{min}, \bm x_\mathrm{max}]$. Raisuddin and De solved this issue by performing pre-iterations that are aimed at calibrating these bounds~\cite{raisuddin2022feqa}, but it is difficult to assess if such approaches work in the general case. Instead, one can solve the linear system by mapping it into a generalized eigenvalue problem whose ground state encodes the linear system solution~\cite{kudo2024annealing}. Using this approach, bounds have to be prescribed on the continuous variables of the eigenvectors that are defined up to a multiplicative constant. 
As eigenvectors are defined as such, $\bm x_\mathrm{min}$ and $\bm x_\mathrm{max}$ can be arbitrarily fixed to $-\bm 1$ and $\bm 1$ and no prior information is required anymore. In general, it is not possible to have prior informations about the order of magnitude of the Helmholtz problem solution as it is dependent to the wavenumber and to the source term. Hence, we decide to solve the gEVPs using AQAE~\cite{a20212} instead of directly solving the linear systems.

\section{Results}
In this section, the one-dimensional Helmholtz problem is solved using both AQAE and ACAE. First, we compute the eigenmodes and eigenfrequencies by solving the generalized eigenvalue problem \eqref{eq:homogeneous_gevp} arising from the finite element discretization of \eqref{eq:homogeneous_problem}. Afterwards, the response of the system to a source term will be computed by solving the linear system \eqref{eq:helmholtz_normal}, later expressing in its gEVP formulation \eqref{eq:helmholtz_normal_gevp}, arising from the finite element discretization of \eqref{eq:hermitian_problem} for various wavenumbers. The success of AQAE and ACAE for computing the eigenvalues and the eigenmodes of a generalized eigenvalue problem $H\bm \phi = \lambda M \bm \phi$ is quantified by the evolution of these three metrics:
\begin{enumerate}
    \item[i)] the relative residual, expressed as $R(i)/R(0)$ where $R(i) = ||H \hat{\bm\phi}^{(i)} - \lambda^{(i)} M  \hat{\bm\phi}^{(i)}||^2$;
    \item[ii)] the relative eigenvalue discrepancy, expressed as $|\lambda^{(i)} - \overline{\lambda}|/|\overline{\lambda}|$ where $\overline{\lambda}$ is a classical algorithm's reference eigenvalue;
    \item[iii)] the relative eigenmode discrepancy, expressed as $|\bm \phi^{(i)} - \overline{\bm \phi}|/|\overline{\bm \phi}|$ where $\overline{\bm\phi}$ is a classical algorithm's reference eigenmode;
\end{enumerate}
where the index $i$ denotes the box algorithm iteration number.

The impact of parameters such as the number of samples and the annealing time on the classical/quantum annealing procedures has already been studied in several recent publications (\textit{e.g.}, Ref.~\cite{nguyen2024quantum}). Here, we will study how iterative quantum annealing-based algorithms such as AQAE behave with respect to the condition number $\sigma $ (i.e., the ratio between the largest and the smallest eigenvalue of the problem) and the binary discretization dictated by the parameter $D$.

\subsection{Homogeneous problem}\label{sec:homogeneous_results}
The homogeneous Helmholtz problem, discretized into the generalized eigenvalue problem \eqref{eq:homogeneous_gevp}, is solved using AQAE and ACAE for the three first eigenmodes and eigenfrequencies and results are shown in Figure~\ref{fig:pannel_gevp}. Results show that the ACAE algorithm converged towards the first three eigenmodes, while only the first eigenmode is successfully computed using AQAE with DWave's quantum annealer~\cite{dwavedocumentation}.
Here, we only show the relative residual as it is the only metric that is independent from the ground truth. The convergence of the eigenvalue and eigenmode discrepancies is shown in App.~\ref{app:complementary}. As expected, the relative residual convergence curves coincide with the graphical results of the eigenmodes. We notice that AQAE (performed using DWave's quantum annealer) converged for the computation of the first eigenmode, but failed for the second and third eigenmodes. ACAE converged for both of the three eigenmodes.

\begin{figure*}[htbp]
    \centering
    \input{pannel_gevp_new.tikz}
    \caption{\textbf{Solution of the one-dimensional homogeneous Helmholtz problem} \eqref{eq:homogeneous_gevp}. Results obtained using ACAE (disks and stars) and AQAE (crosses and pluses) for the three first eigenmodes (labelled with $n = 0,1,2$), for first order ($p=1$) and fifth order ($p=5$) finite elements, with $D=2$. The dashed lines represent the numerical reference solution obtained via a classical solver on a refined mesh. 
    Other parameters are: $N_\delta=25$, $N_\lambda=10$, $r=0.5$,  \texttt{numreads$=100$} (AQAE) and 1000 (ACAE), $t_\mathrm{a} = 100\,\mu$s, and \texttt{beta\_schedule}=geometric~\cite{dwavedocumentation}.
    }
    \label{fig:pannel_gevp}
\end{figure*}

\subsection{Non-homogeneous problem}\label{sec:non_homogeneous_results}
The non-homogeneous Helmholtz problem, discretized and expressed as the generalized eigenvalue problem \eqref{eq:helmholtz_normal_gevp}, is solved using AQAE and ACAE. The solution of \eqref{eq:helmholtz_normal} is derived from the ground-state of \eqref{eq:helmholtz_normal_gevp}.  The finite element order has been varied between $p=1$ and $p=5$, the density of the binary discretization has been varied in $D\in\{2,3,4,5,6,7,8\}$ and the vacuum wavenumber have been varied in $k_0\in\{0, \pi, 2\pi\}$. Results are shown in Figure~\ref{fig:pannel_helmholtz}. Here again, only the relative residual is shown in the main text, while the eigenvalue and eigenmode discrepancies are shown in App.~\ref{app:complementary} as a complementary resource.
As classical computational resources are inexpansive compared to quantum ones, \texttt{numreads} is fixed to 10000 when simulated annealing is chosen while it is fixed to 1000 for quantum annealing. By doing so, ACAE results are considered as a reference. The results of the ACAE show that the success of the algorithm depends on $D$. Table~\ref{tab:K_helmholtz} shows, for all the cases that have been studied, the minimum number of binary variables $D^\star$ that led to the convergence of ACAE. Unlike for the homogeneous Helmholtz problem, AQAE (with $D=D^\star$) did not converge in any of our test cases.

\begin{figure*}[htbp]
    \centering
    \input{pannel_helmholtz_new.tikz}
    \caption{\textbf{Solution of the one-dimensional non-homogeneous Helmholtz problem} \eqref{eq:helmholtz_normal}. Results obtained using AQAE (crosses and pluses) and ACAE (disks and stars), for first order ($p=1$) and fifth order ($p=5$) finite elements with (a) $k_0=0$, (b) $k_0=\pi$, and (c) $k_0=2\pi$. AQAE results are generated with $D=D^\star$ listed in Table~\ref{tab:K_helmholtz}. Thin dashed lines correspond to the finite element solution with the same ($N,p$) combinations as for AQAE and ACAE. Error bars correspond to the standard deviation of the measurements. 
    Other parameters are: $N_\delta=25$, $N_\lambda=10$, $r=0.5$,  \texttt{numreads$=1000$} (AQAE) and 10000 (ACAE), $t_\mathrm{a} = 100\,\mu$s, and \texttt{beta\_schedule}=geometric~\cite{dwavedocumentation}.}
\label{fig:pannel_helmholtz}
\end{figure*}

\begin{table}[h!]
    \centering
    \begin{tabular}{|c|c|c|c|}
        \hline
         & $k_0=0$ & $k_0=\pi$ & $k_0=2\pi$ \\ \hline
        $N=10$, $p=1$& $D^\star = 5$ & $D^\star = 6$ & $D^\star = 3$ \\ \hline
        $N=2$, $p=5$ & $D^\star = 6$ & $D^\star = 7$ & $D^\star = 3$ \\ \hline
    \end{tabular}
    \caption{Minimum binary discretization $D^\star$ that led to convergence for solving \eqref{eq:helmholtz_normal_gevp} using ACAE.}
    \label{tab:K_helmholtz}
\end{table}

\subsection{Discussion}
In this subsection, we review the two key parameters that we believe to influence the convergence of AQAE. There are numerous types of sources of errors in AQAE : i) errors during the annealing procedure, e.g., noise in quantum annealers and insufficient annealing times (in this case, annealers do not find the optimal solution of the QUBOs), ii) errors in the problem representation (in this case, annealers do find the optimal solution but of the wrong QUBOs), iii) errors due to incorrect path in the optimal solutions (in this case, the sequence of optimal solutions does not converge to the global minimum of the continuous objective function). The influence of the annealing time $t_\mathrm{a}$ and of the number of samples \texttt{numreads} have been studied by many researchers \cite{nguyen2024quantum}. In this paper, we investigate how other sources matter in QA-based methods. 

\subsubsection{System condition number}
We observe that the condition number $\sigma$ of the linear system is strongly correlated with $D^\star$ (see Table~\ref{tab:sigma_helmholtz} and compare to Table~\ref{tab:K_helmholtz}). Hence, it seems that at fixed $D$, the spectral properties of the system dictate the convergence of AQAE and ACAE optimal path (i.e., the sequence of optimal solutions) towards the global optimum of a continuous objective function.
\begin{table}[h!]
    \centering
    \begin{tabular}{|c|c|c|c|}
        \hline
         & $k_0=0$ & $k_0=\pi$ & $k_0=2\pi$ \\ \hline
        $N=10$, $p=1$& $\sigma = 175^2$ & $\sigma= 382^2$ & $\sigma = 34^2$ \\ \hline
        $N=2$, $p=5$ & $\sigma = 979^2$ & $\sigma = 4801^2$ & $\sigma = 133^2$ \\ \hline
    \end{tabular}
    \caption{Condition number of $(K+M)^\dag(K+M)$ (with $K$ and $M$ defined in \eqref{eq:linear_system}) for $(N,p) \in \{(10,1),(2,5)\}$ and $k_0 \in \{0,\pi,2\pi\}$.}
    \label{tab:sigma_helmholtz}
\end{table}

\paragraph{Interpretation} Consider the linear system $A\bm x = \bm b$ with Hermitian and positive definite $A$. This linear system can be formulated as the following convex optimization problem
\begin{equation}\label{eq:convexoptimization}
    \bm x_\star = \arg \min_{\bm x}\frac{1}{2}\bm x^T A \bm x - \bm b^T\bm x.
\end{equation}
Let us solve this optimization problem by discretizing the space with a finite number of discrete variables, similarly as AQAE. Let us note $\Tilde{\bm x}_\star$ the approximate optimal solution, that is, the best solution among the finite set of discrete variables.
If the condition number of $A$ is one ($\sigma(A) = 1$), the approximate optimal solution is the one which minimizes the Euclidian distance $||\Tilde{\bm x}_\star - \bm x_\star||$. Hence, in this case, the box-algorithm will always refine the grid around the closest approximate solution and the AQAE will converge for all $D$. If the condition number of $A$ is larger than one ($\sigma(A) > 1$), the approximate optimal solution is the one that minimizes the distance $||\Tilde{\bm x}_\star - \bm x_\star||_A = [(\tilde{\bm{x}}_\star - \bm{x}_\star)^\dag A(\tilde{\bm{x}}_\star - \bm{x}_\star)]^{1/2}$ and not the Euclidian distance. Consequently, it may happen that the approximate solution is not the closest to the global optimum, which would make the box-algorithm zoom into the wrong region of the space. This is illustrated in Figure~\ref{fig:pannel_condition_number}. In Figure~\ref{fig:pannel_condition_number}a, $\sigma(A)$ is small and the box algorithm successfully refines the approximate solution close to the global optimum at each iteration. In Figure~\ref{fig:pannel_condition_number}b, $\sigma(A)$ is large and the box algorithm refines around an approximate solution that is not the closest to the global optimum, and finally gets stuck in a non-optimal region. For large condition number, this issue can be handled by either increasing the number of discrete variables or decreasing the geometric ratio.

\begin{figure*}[htbp]
    \centering
    \input{pannel_boxalgo.tikz}
    \caption{\textbf{Box algorithm examples.} Results obtained by solving \eqref{eq:convexoptimization} with $D=2$ and $r = 0.5$ for (a) $\sigma(A) \approx 5$, (b) $\sigma(A) \approx 50$. The discrete variables are represented as black points, the approximate optimum is represented as a black star, and the global optimum is represented as a blue star.}
    \label{fig:pannel_condition_number}
\end{figure*}

\subsubsection{Integrated control errors}
As mentionned earlier, there are different kinds of error sources in D-Wave's quantum annealers, such as errors due to the fact that the process is non adiabatic in practice (the annealing time being finite and often too small), or errors due to noise in the quantum processor~\cite{Oshiyama2022, Gulacsi2023}. Here, we believe that integrated control errors (ICE) \cite{dwavedocumentation}, i.e., errors in the mapping between the QUBO and the qubits interactions in the quantum annealer, are the major cause of the algorithm failing. We also believe that the fact that ICE lead to convergence issues or not is problem dependent. In order to evaluate how ICE can lead to convergence issues, we model ICE as in Ref.~\cite{teplukhin2019calculation} by a simple additional gaussian contribution in the QUBOs
\begin{equation*}
    Q_\mathrm{eff} = Q + \delta Q,\quad \delta Q_{ij} = \eta_{ij} |Q_\mathrm{max}|,
\end{equation*}
where $\eta_{ij} \sim \mathcal{N}(0, \sigma_\eta)$ and $|Q_\mathrm{max}|$ is the maximum value of $Q$ in absolute terms. We later solve the perturbed QUBOs using ACAE. Figure~\ref{fig:pannel_perturbation_gevp} shows the amount of orders of magnitude that separate the relative residual at the first and last ACAE iteration for the computation of the three first eigenmodes (i.e., the results of Sec.~\ref{sec:homogeneous_results}), with respect to $\sigma_\eta$ that quantifies the magnitude of ICE. Therefore, it quantifies how the magnitude of ICE affects the convergence of ACAE. We observe that above a certain threshold $\sigma_\eta > \sigma^\star$ that is specific to the problem we solve, ACAE does not converge anymore. We refer to this threshold as the critical ICE magnitude. We define it as the ICE for which the relative residual saturated above $10^{-5}$, on average.
\begin{figure}[htbp]
    \centering
    \input{pannel_perturbation_gevp_v2.tikz}
    \caption{Relative residual as a function of the perturbation $\sigma_\eta$ (a) in a homogeneous material for first order finite elements ($N=10$, $p=1$) (b) for fifth order finite elements ($N=2$, $p=5$). Error bars correspond to the standard deviation of the measurements. ACAE parameters: $D=2$, $N_\delta=25$, $N_\lambda=10$, $r=0.5$ \texttt{numreads}$=1000$, \texttt{beta\_schedule}=geometric~\cite{dwavedocumentation}.}
    \label{fig:pannel_perturbation_gevp}
\end{figure}
For the first eigenmode, we find $\sigma_\eta^\star \approx 0.008$ (cf. Figure~\ref{fig:pannel_perturbation_gevp}~a) and $\sigma_\eta^\star \approx 0.006$ (cf. Figure~\ref{fig:pannel_perturbation_gevp}~b) for first-order and fifth-order finite element, respectively. The critical ICE is significantly lower for the higher modes ($n=1,2$) and is therefore problem-dependent. Similarly, Figure~\ref{fig:pannel_perturbation_helmholtz} shows how ICE affects the convergence of ACAE for the non-homogeneous problem solved in Sec.~\ref{sec:non_homogeneous_results}. It can be observed that a ICE magnitude lower of 0.2$\%$ is enough to make ACAE not converge for all the cases that have been studied.
In D-Wave's quantum annealer, the ICE magnitude is roughly of the order of 1$\%$ \cite{dwavedocumentation}. This therefore explains why AQAE did not converge in Sec.~\ref{sec:homogeneous_results} for the second and the third modes, and in Sec.~\ref{sec:non_homogeneous_results} in all the cases. It is worth noting that hardware ICE are expected to decrease in the future, and can be handled using quantum annealing correction algorithms~\cite{QAC, pudenz2015quantum}.

\begin{figure}[htbp]
    \centering
    \input{pannel_perturbation_helmholtz_v2.tikz}
    \caption{Relative residual as a function of the perturbation $\sigma_\eta$ for $D=D^\star$ and (a) $k_0=0$, (b) $k_0=\pi$ , (c) $k_0=2\pi$. Error bars correspond to the standard deviation of the measurements. ACAE parameters: $N_\delta=25$, $N_\lambda=10$, $r=0.5$ \texttt{numreads}$=10000$, \texttt{beta\_schedule}=geometric~\cite{dwavedocumentation}.}
    \label{fig:pannel_perturbation_helmholtz}
\end{figure}

\subsubsection{Scaling analyses}

Solving Helmholtz problems leads to solving linear systems. Depending on the wavenumber, the linear system is positive-definite or indefinite in the low- and high-frequency regimes, respectively. Such systems can be solved using Krylov subspace iterative methods such as Conjugate Gradients for positive definite systems, or GMRES or CGNR (CG on the normal equations), in the indefinite case. Without preconditioning and in exact arithmetic, these algorithms have a worst complexity of $\mathcal{O}(N_\mathrm{DOF}^2)$. The most competitive classical approaches rely on preconditioning techniques, e.g., multigrid~\cite{elman2001multigrid} or multi-level domain decomposition methods~\cite{gander2002optimized}, which can scale as $\mathcal{O}(N_\mathrm{DOF})$ and $\mathcal{O}(N_{\mathrm{DOF}}^{{(d+1)}/{d}})$ for definite and indefinite systems, respectively, with $d$ the number of space dimensions~\cite{galkowski2025convergence}. For indefinite systems, these scalings hold only if it is assumed that infinite parallel resources are available.

Whether QA-based methods are valuable or not depends on the comparison with the scaling of these classical algorithms. Unlike classical algorithms, though, QA-based methods do not explicitly scale with $N_\mathrm{DOF}$, but rather with the minimum energy gap $\Delta_\mathrm{min}$ of the Hamiltonian $H(t)$ that is related to the annealing time $t_\mathrm{a} = \mathcal{O}(\Delta_\mathrm{min}^{-2})$. The analysis below is relevant for one single AQAE iteration. However, the total number of iterations is independent of the number of degrees of freedom, since it only depends on the final precision $\varepsilon$ as $N_\delta \propto \log_r(1/\varepsilon)$.\\

By definition, the minimum energy gap is upper bounded by the energy gap of the classical Hamiltonian $H_\mathrm{p}$, which we write $\Delta_\mathrm{p}$. The eigenvalues of $H_\mathrm{p}$ are the values of the discretized objective function. As an example, consider the minimisation of the following functional
    \begin{equation}\label{eq:optimizationproblem}
        \mathcal J(x) = \frac{1}{2}\bm x^\dag A \bm x - \frac{1}{2}\bm x^\dag \bm b - \frac{1}{2}\bm b^\dag \bm x.
    \end{equation}
    Mapping the continuous variables $\bm x$ into spins ($\bm x \rightarrow \bm x(\bm s)$) leads to a discrete Ising problem that has as eigenvalue spectrum the ensemble of discrete values of the objective function at each $\bm x(\bm s)$. From this observation, we estimate the value of $\Delta_\mathrm{p}$ as follows:
\begin{align}
    \Delta_\mathrm{p} &= E_1 - E_0 = E_0 + \bm y^\dag A \bm y - E_0 \\
    &= \left({\sum_i c_i \bm\phi_i} \right)^\dag \left(\sum_j \lambda_j {\bm\phi_j}{\bm\phi_j^\dag}\right)\left({\sum_kc_k\bm\phi_k}\right)\\
    &= \sum_i |c_i|^2 \lambda_i, 
\end{align}
with $E_1$ and $E_0$ the two lowest values of the discretized objective function, $\lambda_i$ the $i$-th lowest eigenvalue of $A$ and $\bm \phi_i$ the corresponding eigenvector, and $\varepsilon$ the discretization step. The energy gap of $H_\mathrm{p}$ is therefore related to the spectral properties of $A$. Numerically, we observe that $\Delta_\mathrm p$ grows with $\mathcal{O}(\mathrm{Tr}(A)/N_\mathrm{DOF})$, i.e., $\Delta_\mathrm p$ grows with the mean eigenvalue of $A$ (see Figure~\ref{fig:delta_p}). Since it can be shown that for Helmholtz problems $\mathrm{Tr}(A)$ grows as $\mathcal{O}(N_\mathrm{DOF}^{2/d})$ and $\mathcal{O}(N_\mathrm{DOF}^{4/d-1})$ for the original and normal equations respectively, we obtain the following expected lower bound for the annealing time
\begin{equation}
    t_\mathrm{a} = \mathcal{O}(\Delta_\mathrm{min}^{-2}) \geq \mathcal{O}(\Delta_\mathrm{p}^{-2}) = \begin{cases}
        \mathcal{O}(\varepsilon^{-4} N_\mathrm{DOF}^{2-4/d})\, \text{at LF},\\
        \mathcal{O}(\varepsilon^{-4} N_\mathrm{DOF}^{4-8/d})\, \text{at HF}.
    \end{cases}
\end{equation}

\begin{figure}[htbp]
    \centering
    \input{delta_p.tikz}
    \caption{\textbf{Evolution of the Ising gap with the number of degrees of freedom}. The Ising gap encodes the discretized objective function \eqref{eq:optimizationproblem} with $D=1$, $p=1$ and $\varepsilon = 1$. Dashed lines show the evolution of $\Delta_\mathrm p$ with the number of degrees of freedom for the original equations (in blue) and the normal equations (in red). Solid lines represents the evolution of the mean eigenvalue of the original equations (in blue) and the normal equations (in red).}
    \label{fig:delta_p}
\end{figure}
\noindent These lower bounds suggest that a quantum advantage, if it exists, could occur only for $d \leq 4$ and $d \leq 3$ at LF and HF, respectively.

Note finally that $\Delta_\mathrm{min}$ does not generally scale polynomially, and can scale exponentially due to discontinuous quantum phase transitions. Different strategies have been investigated to avoid such quantum phase transitions, e.g., by modifying the driver Hamiltonian $H_0$~\cite{seki2012quantum}, but those methods are not implemented yet in D-Wave's quantum annealer.
In order to estimate the behavior of the gap of the full Hamiltonian (not only $H_\mathrm{p}$), here we perform finite-size scaling for $\Delta_\mathrm{min}$ by simulating $\Delta(s)$ for various $N_\mathrm{DOF}$ for the normal equations arising from the first-order finite element discretization of the Helmholtz problem we solved previously. We considered three different Hamiltonians :
\begin{enumerate}
    \item[(i)] The D-Wave's Hamiltonian
    \begin{equation}
        H_\mathrm{DWave}(s) = A(s) H_0 + B(s)\tilde H_\mathrm{p},
    \end{equation}
    with $A(s)$ and $B(s)$ available in D-Wave's documentation~\cite{dwavedocumentation}, $H_0 = \sum_i \sigma_x^{(i)}$, and $\tilde H_\mathrm p = \sum_i \tilde{h}_i \sigma_z^{(i)} + \sum_{i > j} \tilde{J}_{ij} \sigma_{z}^{(i)} \sigma_{z}^{(j)}$ the Ising Hamiltonian which encodes the normal equations arising from the Helmholtz problem, rescaled according to $(J_{ij}, h_j)\rightarrow (\tilde{J}_{ij}, \tilde{h}_j) = (J_{ij}, h_j)/\max (|J|,|\bm h|)$ in order to reproduce the normalization applied in D-Wave’s quantum processing unit, which rescales problem coefficients to fit within its programmable range, and where $\sigma_x$ and $\sigma_z$ are the usual Pauli operators.\\
    \item[(ii)] The linear scheduling Hamiltonian
    \begin{align}
        H_\mathrm{lin}(s) &= A(s) H_0 + B(s) H_\mathrm{p},
    \end{align}
    with $A(s) = 1-s$, $B(s) = s$ and $H_\mathrm p$ the Ising Hamiltonian which encodes the normal equations arising from the Helmholtz problem.\\
    \item[(iii)] The linear scheduling Hamiltonian with antiferromagnetic fluctuations~\cite{seki2012quantum}
    \begin{equation}
        H_\mathrm{af}(s) = A(s)H_0 + B(s)\left( A(s) H_1 + B(s) H_\mathrm{p} \right),
    \end{equation}
    with $A(s) = 1-s$, $B(s) = s$, $H_0 = \sum_i \sigma_x^{(i)}$ and $H_1 = N\left( \sum_i\sigma_x^{(i)}/N \right)^2$.\\
\end{enumerate}

\begin{figure*}[ht]
    \centering
    \input{pannel_gap.tikz}
    \caption{\textbf{Time evolution of the Hamiltonian gap for various problem sizes.}
    Results obtained using $D=1$, $p=1$ and $\varepsilon=1$, with (a) $H(s) = H_\mathrm{DWave}$, (b) $H(s) = H_\mathrm{lin}$, (c) $H(s) =  H_\mathrm{af}$. The wave number is varied between $k_0 = 0$ (left column), $k_0=\pi$ (middle column) and $k_0 = 2\pi$ (right column).}
    \label{fig:enter-label}
\end{figure*}

\begin{figure*}[ht]
    \centering
    \input{scalings.tikz}
    \caption{\textbf{Scaling of the minimum Hamiltonian gap with the number of degrees of freedom.} Results obtained using $D=1$, $p=1$ and $\varepsilon=1$, for  (a) $H(s) = H_\mathrm{DWave}$, (b) $H(s) = H_\mathrm{lin}$, (c) $H(s) =  H_\mathrm{af}$.}
    \label{fig:scalings}
\end{figure*}

Figure~\ref{fig:enter-label}a-b-c shows the time evolution of the Hamiltonian gap of $H_\mathrm{DWave}$, $H_\mathrm{lin}$, and $H_\mathrm{af}$, respectively. The left, middle and right columns correspond to $k_0 = 0$, $k_0 = \pi$ and $k_0 = 2\pi$, respectively. The evolution of the minimum Hamiltonian gap with the number of degrees of freedom is shown in Figure~\ref{fig:scalings} for all $(H(s),k_0)$ combinations.\\

For $k_0 = 0$, we observe that the minimum gap scales as $\mathcal{O}(N_\mathrm{DOF}^{-0.83})$ and $\mathcal{O}(N_\mathrm{DOF}^{-0.8})$ for $H_\mathrm{DWave}$ and $H_\mathrm{lin}$, which leads to $t_\mathrm{a} = \mathcal{O}(N_\mathrm{DOF}^{1.66})$ and $t_\mathrm{a} = \mathcal{O}(N_\mathrm{DOF}^{1.6})$, respectively. The gap closing seems to be slower for $H_\mathrm{af}$, but finite size effects are too strong to extrapolate.\\

\noindent For $k_0 = \pi$, we observe that the minimum gap decays as $\mathcal{O}(N_\mathrm{DOF}^{-1.36})$, $\mathcal{O}(N_\mathrm{DOF}^{-1.16})$ and $\mathcal{O}(N_\mathrm{DOF}^{-0.91})$ for $H_\mathrm{DWave}$, $H_\mathrm{lin}$ and $H_\mathrm{af}$, respectively. This leads to $t_\mathrm{a} = \mathcal{O}(N_\mathrm{DOF}^{2.72})$, $t_\mathrm{a} = \mathcal{O}(N_\mathrm{DOF}^{2.32})$ and $t_\mathrm{a} = \mathcal{O}(N_\mathrm{DOF}^{1.82})$, respectively.\\

\noindent For $k_0 = 2\pi$, we observe the minimal gap scalings $\mathcal{O}(N_\mathrm{DOF}^{-1.18})$, $\mathcal{O}(N_\mathrm{DOF}^{-1.07})$ and $\mathcal{O}(N_\mathrm{DOF}^{-0.76})$ for $H_\mathrm{DWave}$, $H_\mathrm{lin}$ and $H_\mathrm{af}$, respectively, leading to $t_\mathrm{a} = \mathcal{O}(N_\mathrm{DOF}^{2.36})$, $t_\mathrm{a} = \mathcal{O}(N_\mathrm{DOF}^{2.14})$ and $t_\mathrm{a} = \mathcal{O}(N_\mathrm{DOF}^{1.52})$.\\

\noindent These scalings of the annealing time are summarized in Table~\ref{tab:scalings}. We notice that the minimum gap of $H_\mathrm{lin}$ closes slower than the one of $H_\mathrm{DWave}$ for the three problems we solved. Antiferromagnetic fluctuation terms in the driving did slowen the gap closing even more for $k_0=\pi$ and $k_0=2\pi$, at least. This gives insights into the possibility to engineer annealing scheduling and driving Hamiltonians to obtain favorable annealing time scalings. Finally, we note that the annealing-based approach seems to outperform naive classical algorithms such as unpreconditioned CG and GMRES, which scale as $\mathcal{O}(N_\mathrm{DOF}^2)$, for $(H=H_\mathrm{DWave},k_0 = 0)$, $(H=H_\mathrm{lin},k_0 = 0)$, $(H=H_\mathrm{af},k_0 = \pi)$ and $(H=H_\mathrm{af},k_0 = 2\pi)$.

\begin{table}[h!]
    \centering
    \begin{tabular}{|c|c|c|c|}
        \hline
         & $k_0=0$ & $k_0=\pi$ & $k_0=2\pi$ \\ \hline
        $H_\mathrm{DWave}$         & $\mathcal{O}(N_\mathrm{DOF}^{1.66})$ & $\mathcal{O}(N_\mathrm{DOF}^{2.72})$ & $\mathcal{O}(N_\mathrm{DOF}^{2.36})$ \\ \hline
        $H_\mathrm{lin}$           & $\mathcal{O}(N_\mathrm{DOF}^{1.60})$ & $\mathcal{O}(N_\mathrm{DOF}^{2.32})$ & $\mathcal{O}(N_\mathrm{DOF}^{2.14})$ \\ \hline
        $H_\mathrm{af}$            & $-$                                  & $\mathcal{O}(N_\mathrm{DOF}^{1.82})$ & $\mathcal{O}(N_\mathrm{DOF}^{1.52})$ \\ \hline
    \end{tabular}
    \caption{\textbf{Annealing time scalings for the Helmholtz problem.} Results obtained using $D=1$, $p=1$ and $\varepsilon=1$. The annealing time is considered proportional to $\Delta_\mathrm{min}^{-2}$, with $\Delta_\mathrm{min}$ available in Figure~\ref{fig:scalings} for each $(H,k_0)$ combinations.}
    \label{tab:scalings}
\end{table}

\section{Conclusion}

A one-dimensional Helmholtz problem with Dirichlet and Neumann boundary conditions has been solved with adaptive quantum annealing eigensolver (AQAE) and adaptive classical annealing eigensolver (ACAE) algorithm. The homogeneous problem has been solved at first to find the eigenmodes of a vacuum-Si$\mathrm{O}_2$ medium. A source term has been added to compute the response for various wavenumbers. Results obtained using a classical annealer have been considered as a ground truth and have highlighted that large condition numbers lead to non-convergence of the algorithms. To address this, the density of discretization (dictated by the parameter
$D$, i.e., the number of binary variables used to approximate continuous variables) must be refined accordingly. This result is valid for all box-algorithm-based methods. The AQAE results, obtained using D-Wave's quantum annealer, have highlighted the current limitations of quantum annealing hardware in terms of integrated control errors (ICE). The results have also shown that the robustness of the AQAE algorithm with respect to ICE is problem-dependent. We also discussed, based on gap scaling analyses, the possibility of the algorithm to lead to a quantum advantage and concluded that it depends on both the frequency regime and the dimension of the problem, if discontinuous quantum phase transitions are avoided. The potential impact of the driving Hamiltonian has also be highlighted, giving insights into the possibility to engineer driving Hamiltonians to obtain favorable gap closings. As a perspective, non-hermitian systems (e.g., Helmholtz problems with loss) could be considered, and a more detailed finite-size scaling of the performance of the algorithm with respect to different noise sources could be realized, in order to determine e.g., which noise source is the most detrimental. QA could also be combined with classical Helmholtz preconditioning techniques, with the potential to further improve overall computational complexity.

\section*{Data availability}
All our results can be reproduced using our codes in Ref.~\cite{remi2025github}.

\section*{Acknowledgement}
This work was funded in part by the European Regional Development Fund (ERDF) and the Walloon Region through project 1149 VirtualLab\_ULiege (program 2021-2027). This work is partially supported by a ‘‘Strategic Opportunity’’ grant from the University of Liège.

\begin{appendix}

\section{Mathematical developments}
\label{app:maths}

\subsection{Min-max theorem for gEVPs}
\label{app:minmaxthm}

In this appendix, the min-max theorem for eigenvalue problems will be derived and generalized to generalized eigenvalue problems.

\begin{lemma}[]
    Let $H$ be a Hermitian operator that belongs to a Hilbert space $\mathcal{H}$ with eigenvalues ordered such that $\lambda_0 \leq \cdots \leq \lambda_{N-1}$ and corresponding eigenvectors $\bm \phi_0, \cdots, \bm \phi_{N-1}$. Let $\mathcal{S}_{k} \subseteq \mathcal{H}\setminus \{\bm 0\}$ with $\dim(\mathcal{S}_{k}) = N-k$, then
    \[
    \exists \, \bm \phi \in \mathcal{S}_{k}\quad \mathrm{s.t.}\quad \frac{\bm \phi^\dag H \bm \phi}{\bm \phi^\dag \bm \phi} \leq \lambda_k.
    \]
\end{lemma}
\begin{proof}
    Dimensional reasonings allows to write
    \begin{align*}
    \dim(\mathcal{S}_{k}) + \dim(\mathrm{span}(\bm \phi_0, \cdots, \bm \phi_{k})) &= N+1 > N\\ \Rightarrow \mathcal{S}_{k} \cap \mathrm{span}(\bm \phi_0, \cdots, \bm \phi_{k}) &\neq \emptyset.
    \end{align*}
    Hence, for any $\bm \phi \in \mathcal{S}_{k} \cap \mathrm{span}(\bm \phi_0, \cdots, \bm \phi_{k})$, it is trivial that the lemma is verified.
\end{proof}

\begin{theorem}[Min-Max Theorem]
    Let \( H \) and \( M \) be Hermitian matrices in $\mathbb{C}^{N_{\mathrm{DOF}}\times N_{\mathrm{DOF}}}$, with \( M \) being positive definite. Let us define a Hilbert space $\mathcal{S}_{k}\subseteq \mathbb{C}^{N_{\mathrm{DOF}}\times N_{\mathrm{DOF}}}\setminus\{\bm 0\}$ with $\dim(\mathcal{S}_k) = N_{\mathrm{DOF}}-k$. Then, the generalized eigenvalues \( \{\lambda_i\} \) (with $\lambda_i$ ordered such that $\lambda_0 \leq \cdots \leq \lambda_{N_{\mathrm{DOF}}-1}$) of the generalized eigenvalue problem
    \[
    H \bm{\phi} = \lambda M \bm{\phi}
    \]
    are given by
    \[
    \lambda_k = \displaystyle\max_{\mathcal{S}_{k}} \, \displaystyle\min_{\bm{\phi} \in \mathcal{S}_{k}} R_{H,M}(\bm \phi) = \displaystyle\max_{\mathcal{S}_{k}} \, \displaystyle\min_{\bm{\phi} \in \mathcal{S}_{k}} \frac{\bm{\phi}^\dag H \bm{\phi}}{\bm{\phi}^\dag M \bm{\phi}},\\
    \]
    where $R_{H,M}$ denotes the Rayleigh quotient for the ($H,M$) pair, and $\bm \phi^\dag$ denotes the conjugate transpose of $\bm \phi$.
\end{theorem}

\begin{proof}
Consider the standard EVP
\[
H \bm \phi = \lambda \bm \phi,
\]
for Hermitian $H$. Using the above lemma,
\[
    \displaystyle\min_{\bm{\phi} \in \mathcal{S}_{k}} \frac{\bm{\phi}^\dag H \bm{\phi}}{\bm{\phi}^\dag \bm{\phi}} \leq \lambda_k.
\]
Taking $\mathcal{S}_{k} = \mathrm{span}(\bm \phi_k, \cdots, \bm \phi_{N_{\mathrm{DOF}}-1})$ gives
\[
    \displaystyle\min_{\bm{\phi} \in \mathcal{S}_{k}} \frac{\bm{\phi}^\dag H \bm{\phi}}{\bm{\phi}^\dag \bm{\phi}} \geq \lambda_k,
\]
which concludes the proof. This can be generalized for
\[
H \bm \phi = \lambda M \bm \phi
\]
with Hermitian positive definite $M$. Taking the Cholesky decomposition $M = L^\dag L$, the variational functional becomes
\begin{align*}
    \frac{\bm{\phi}^\dag H \bm{\phi}}{\bm{\phi}^\dag M \bm{\phi}} = \frac{\bm{\phi}^\dag H \bm{\phi}}{\bm{\phi}^\dag L^\dag L \bm{\phi}}
    = \frac{\bm{\phi}^\dag L^\dag L^{-\dag}HL^{-1}L \bm{\phi}}{\bm{\phi}^\dag L^\dag L \bm{\phi}} = \frac{\bm y^\dag \Tilde{H}\bm y}{\bm y^\dag \bm y},
\end{align*}
where $\Tilde{H} = L^{-\dag}HL^{-1}$ and $\bm y = L \bm \phi$. This shows the equivalence between the variational formulation of standard and generalized EVPs.
\end{proof}

\subsection{Continuous to binary quadratic functions mapping}
\label{app:gevpqubomapping}

Once a continuous-to-binary mapping $\phi_j \rightarrow \phi_{j,\mathrm{min}} + \bm b^T \bm q_j$ is defined, any quadratic expression $\bm \phi^T A \bm \phi$, $\bm \phi \in \mathbb{R}^{N_{\mathrm{DOF}}}$ can be written as a quadratic binary expression as follows:

\begin{widetext}
\begin{align*}
    \bm \phi^T A \bm \phi \approx \,
    &\bm \phi_\mathrm{min}^T A \bm \phi_\mathrm{min} +
    \begin{bmatrix}
       \bm{q}_0^{T}\bm{b}\cdots \bm{q}_{N_{\mathrm{DOF}}-1}^{T}\bm{b}
    \end{bmatrix}
    \,A   \bm \phi_\mathrm{min}
    +
    \bm \phi_\mathrm{min}^T A
    \begin{bmatrix}
           \bm{b}^T\bm{q}_0\\
           \vdots \\
           \bm{b}^T\bm{q}_{N_{\mathrm{DOF}}-1}
    \end{bmatrix}
    +
    \begin{bmatrix}
       \bm{q}_0^{T}\bm{b}\cdots \bm{q}_{N_{\mathrm{DOF}}-1}^{T}\bm{b}
    \end{bmatrix}
    \,A
    \begin{bmatrix}
           \bm{b}^T\bm{q}_0\\
           \vdots \\
           \bm{b}^T\bm{q}_{N_{\mathrm{DOF}}-1}
    \end{bmatrix}\\
    = \, &\bm \phi_\mathrm{min}^T A \bm \phi_\mathrm{min} + \bm q^T
    \begin{bmatrix}
    D_{1} &        & 0   \\
          & \ddots &     \\
    0     &        & D_{N_{\mathrm{DOF}}}
    \end{bmatrix} \bm q
    +
    \bm q^T
    \begin{bmatrix}
    A_{11}\bm{b}\bm{b}^T & \cdots  & A_{1N_{\mathrm{DOF}}}\bm{b}\bm{b}^T\\
    \vdots & \ddots & \vdots\\
    A_{N_{\mathrm{DOF}}1}\bm{b}\bm{b}^T & \cdots  & A_{N_{\mathrm{DOF}}N_{\mathrm{DOF}}}\bm{b}\bm{b}^T
    \end{bmatrix}
    \bm q,
\end{align*}
\end{widetext}

\noindent where
\[
\bm q = \begin{bmatrix}
           \bm{q}_0 \\
           \vdots \\
           \bm{q}_{N_{\mathrm{DOF}}-1}
        \end{bmatrix},\quad
 D_i = \mathrm{diag}\left( \left(\left(A_{i:}+A_{:i}^T\right)\cdot \bm \phi_\mathrm{min}\right) \bm b\right).
\]
By defining
\begin{equation*}
Q_A =
    \begin{bmatrix}
    D_{1} &        & 0   \\
          & \ddots &     \\
    0     &        & D_{N_{\mathrm{DOF}}}
    \end{bmatrix}
    +
    \begin{bmatrix}
    A_{11}\bm{b}\bm{b}^T & \cdots  & A_{1N_{\mathrm{DOF}}}\bm{b}\bm{b}^T\\
    \vdots & \ddots & \vdots\\
    A_{N_{\mathrm{DOF}}1}\bm{b}\bm{b}^T & \cdots  & A_{N_{\mathrm{DOF}}N_{\mathrm{DOF}}}\bm{b}\bm{b}^T
    \end{bmatrix},
\end{equation*}
we can write the continuous quadratic function as
\begin{equation}
    \bm{\phi}^T A \bm{\phi} \approx \bm \phi_\mathrm{min}^T A \bm \phi_\mathrm{min} + \bm q^T Q_A \bm q.
    \label{eq:binaryquadraticmapping}
\end{equation}
For the computation of the $n$-th eigenmode, the variational function is of the form
\[
f_n(\bm \phi) = \bm{\phi}^\dag H \bm{\phi} - \gamma \bm{\phi}^\dag M \bm{\phi} +\xi_n(\bm \phi),
\]
where $\xi_n(\bm \phi)$ is a quadratic penalty term associated with any $\bm \phi$ that are not $M$-orthogonal to the $n-1$ previous eigenmodes:
\begin{equation}
    \xi_n(\bm \phi) = \sum_{m=0}^{n-1} \beta_m  \bm \phi^T P_{M,\bm \phi_m} \bm \phi,
\end{equation}
with $P_{M,\bm \phi_m} = M \bm \phi_m \bm \phi_m^\dag M^\dag$
and $\beta_m > \lambda_n - \lambda_m$. Using \eqref{eq:binaryquadraticmapping} to map the continuous variables to binary ones, the variational function becomes
\begin{align}
f_n(\bm q) &= \bm \phi_\mathrm{min}^T H \bm \phi_\mathrm{min} + \bm{q}^T Q_H \bm{q} \\&\quad - \gamma \left(\bm \phi_\mathrm{min}^T M \bm \phi_\mathrm{min} + \bm{q}^T Q_M \bm{q} \right) \\&\quad + \sum_{m=0}^{n-1} \beta_m \left( \bm \phi_\mathrm{min}^T P_{M,\bm \phi_m} \bm \phi_\mathrm{min} + \bm q^T Q_{P_{M,\bm \phi_m}} \bm q\right)\label{eq:penalized_cost}.
\end{align}

\section{Complementary results}\label{app:complementary}
In this section, the convergence of the eigenmodes and eigenvalue discrepancy defined in Sec.~\ref{sec:homogeneous_results} are shown in Figure~\ref{fig:complementary_gevp}~-~\ref{fig:complementary_helmholtz} as complements to the results of Figure~\ref{fig:pannel_gevp}~-~\ref{fig:pannel_helmholtz} respectively.

\begin{figure*}[htbp]
    \centering
    \input{complementary_gevp.tikz}
    \caption{Evolution of the eigenmode (top row) and eigenvalue (bottom row) discrepancy with respect to AQAE and ACAE iteration number solving \eqref{eq:homogeneous_gevp}. }
    \label{fig:complementary_gevp}
\end{figure*}

\begin{figure*}[htbp]
    \centering
    \input{complementary_helmholtz.tikz}
    \caption{Evolution of the eigenmode (top row) and eigenvalue (bottom row) discrepancy with respect to AQAE (in red) and ACAE (in viridis color map) iteration number solving \eqref{eq:helmholtz_normal_gevp} with (a) $k_0 = 0$, (b) $k_0 = \pi$, (c) $k_0 = 2\pi$. Error bars represent the standard deviation of the results.}
    \label{fig:complementary_helmholtz}
\end{figure*}

One should keep in mind that the convergence of the residuals do not always imply that the eigenmode computation is successful: if the weight given to the penalization term in \eqref{eq:penalized_cost} is too small, the ground state of the induced Hamiltonian is not orthogonal to the lower eigenmodes of the gEVP.

\end{appendix}

\FloatBarrier
\bibliography{ref}
\end{document}